\documentclass[preprint,12pt]{elsarticle}

\usepackage[utf8]{inputenc}
\usepackage{amsmath}
\usepackage{amsfonts}
\usepackage{amssymb}
\usepackage{amsthm}
\usepackage{graphicx}
\usepackage{mathrsfs,xcolor}
\usepackage{blkarray}
\usepackage{bm}
\usepackage{setspace}
\usepackage[T1]{fontenc}
\usepackage{enumitem}
\usepackage{multirow}
\usepackage{hhline}
\usepackage{adjustbox}
\usepackage[titletoc]{appendix}
\usepackage{float}
\usepackage{url}
\usepackage{subfigure}
\usepackage{threeparttable}
\usepackage{epsfig}
\usepackage{lineno}

\makeatletter
\newcommand\xleftrightarrow[2][]{\ext@arrow 0099{\longleftrightarrowfill@}{#1}{#2}}
\def\longleftrightarrowfill@{\arrowfill@\leftarrow\relbar\rightarrow}
\makeatother

\numberwithin{equation}{section}

\newtheorem{theorem}{Theorem}
\newtheorem{define}{Definition}

\begin{document}

\begin{frontmatter}



\title{Mathematical Modeling of  Soil-Transmitted Helminth Infection: Human-Animal Dynamics with Environmental Reservoirs}

\author[a,b]{Rafiatu Imoro}
\author[b]{Maica Krizna Gavina}
\author[c]{Vachel Gay Paller}
\author[b]{Jomar Rabajante}
\author[b]{Mark Jayson Cortez}
\author[b]{Editha Jose}

\address[a]{Bolgatanga Technical University, Bolgatanga, Upper East Region, Ghana}
\address[b]{Institute of Mathematical Sciences, University of the Philippines Los Ba\~nos, Los Ba\~nos, 4031, Laguna, Philippines}
\address[c]{Institute of Biological Sciences, University of the Philippines Los Baños Los Baños, 4031 Laguna, Philippines}

\begin{abstract}
Soil-transmitted helminth (STH) infections, one of the most prevalent neglected tropical diseases, pose a significant threat to public health in tropical and subtropical areas. These parasites infect humans and animals through direct contact with contaminated soil or accidental ingestion. This study examines the dynamics of STH transmission using a deterministic compartmental model and nonlinear ordinary differential equations. Our model incorporates the roles humans, animals, and the environment play as reservoirs for spreading STH. We derived the basic reproduction number and demonstrate that the disease-free and endemic equilibrium points are asymptotically stable under specific thresholds. We also performed a sensitivity analysis to determine how each parameter affects the model's output. The sensitivity analysis identifies key parameters influencing infection rates, such as ingestion rate, disease progression rate, and shedding rate, all of which increase infection. Conversely, higher clearance and recovery rates decrease infection. The study also highlights the potential for cross-species transmission of STH infections between humans and animals, underscoring the One Health concept, which acknowledges the interdependence of human, animal, and environmental health.
\end{abstract}

\begin{highlights}
\item Developed a mathematical model to show the dynamics of STH transmission with human-animal-environment interactions
\item Ingestion rate, disease progression rate, and shedding rate, increase infection while higher clearance and recovery rates decrease infection
\item The complexity of STH infection dynamics needs multi-faceted approaches to effectively manage and reduce the burden of these infections in both human and animal populations. 
\end{highlights}


\end{frontmatter}




\section{Introduction}

Soil-transmitted helminths (STHs) are a group of parasitic worms that spread through contact or accidental ingestion of contaminated soil, making them among the most widespread infections globally. They present a major public health challenge in tropical and subtropical regions. The three most prevalent STH species that infect humans are roundworms (\emph{Ascaris lumbricoides}), whipworms (\emph{Trichuris trichiura}), and hookworms (\emph{Necator americanus} and \emph{Ancylostoma duodenale}). Over two billion individuals worldwide have been impacted by endemic STH infections, with an additional four billion people at risk of acquiring these infections. The highest burden of STH infections is observed in rural areas of Sub-Saharan Africa, Latin America, China, and Southeast Asia. STH infections are typically acquired by ingesting nematode eggs from contaminated soil, such as those of \emph{A. lumbricoides} and \emph{T. trichiura}, or through the penetration of the skin by larvae present in the soil, as seen with hookworm \cite{bib1, bib2, bib3}. 

STH infections represent a major health burden in areas where they are endemic. Although STH infections do not typically result in high mortality rates, they impose a considerable burden of morbidity, particularly on preschool and school-aged children, hindering their physical and cognitive development and ultimately affecting their academic performance. These infections can also contribute to nutritional deficiencies and anemia \cite{bib4}. In addition, the decreased productivity and efficiency observed in adults due to STH infections is a significant concern \cite{bib5}. The economic impact of STH infections is substantial, with billions of dollars spent annually on interventions \cite{bib6}. Morbidity from STH has been estimated to range between 1.97 and 3.3 million disability-adjusted life years (DALYs) \cite{bib7, bib8}. These infections represent a significant health burden, particularly for individuals living in poverty, despite being both preventable and treatable \cite{bib9}. Current strategies to mitigate the effects of STH infections include administering anthelmintic drugs, like albendazole and mebendazole, to high-risk populations in endemic regions through mass drug administration (MDA) programs \cite{bib10}.

STHs are not exclusive to humans; animals such as dogs, cats, and pigs also serve as hosts to various species of roundworms, whipworms, and hookworms. There is suggestive evidence that these animal STHs can be transmitted to and cause patent infections in humans \cite{bib12, bib11, bib13, bib14, bib15}. The increased density of animal populations close to human dwellings raises the chances of human-animal contact, which in turn raises the risk of zoonotic transmission of STH infections. This risk is particularly high when animals are left unattended, not dewormed, and allowed to defecate indiscriminately \cite{bib16}. Humans can acquire zoonotic STHs similar to human STH species, primarily through contact with the infected eggs or larvae of animal STH species. As STH infections are closely linked to soil contamination, understanding disease transmission at the animal-human-environment interface is becoming increasingly important.

Mathematical models that describe the transmission of infectious disease agents fall into two categories, prevalence models and density models. Prevalence models classify individuals into states such as susceptible, latent, infectious, and immune, and describe the number or proportion of individuals in these states. On the other hand, density models consider the number of parasites within a host and describe the average parasite load \cite{bib17}. Researchers have employed both prevalence \cite{bib18, bib19} and density models \cite{bib20, bib29, bib27, bib22, bib23,  bib25, bib24, bib21} to gain insights into the transmission dynamics of STH infections. The earliest mathematical models for helminth infection, which primarily concentrated on the human host, trace back to \cite{bib30, bib31, bib29, bib28}. 

In this paper, we present a multi-host mathematical model using ordinary differential equations (ODEs), a type of deterministic model, to assess the influence of each host on the transmission dynamics of STH and the long-term survival of helminths in the environment. The model aims to shed light on the contribution of each population to the transmission patterns of STH infections. To the best of our knowledge, no model has yet incorporated the role of animals in the transmission dynamics of STH infections.


The rest of the paper is organized as follows: Section \ref{section: model} introduces the model that captures the interactions between animals, humans, and the environment and their contribution to STH transmission. Section \ref{section:paramestimate} details the estimation of crucial parameters and includes a sensitivity analysis to identify key parameters that significantly impact the model’s output, aiding in making informed decisions for eliminating STH transmission. Section \ref{section:simulation} presents numerical simulations and discusses the model’s outcomes. The paper concludes and presents some future studies in Section \ref{section:conclusion}. Full details of derivation of the reproduction number, and proofs of certain theorems can be found in the Appendix.

\section{Model formulation}\label{section: model}
\subsection{Model assumptions}
\indent

We used the SEIS compartmental model to describe the dynamics of STH infections. Our model considers three populations: humans, animals, and parasitic eggs/larvae. The human and animal populations are split into three subgroups: $ S(t)$, $E(t)$, and $I(t)$. Throughout this paper, state variables and parameters with the subscript $h$ correspond to the human population, while those with the subscript $a$ pertain to the animal population.

$S_h(t)$ and $S_a(t)$ represent the number of individuals who are not infected but can be infected with helminth parasites. $E_h(t)$ and $ E_a(t)$ represent those exposed to helminth infection but do not release parasite eggs. $I_h(t)$ and $ I_a(t)$ correspond to the number of individuals infected with parasitic worms.  $M(t)$ represents the number of parasitic eggs/larvae in the environment capable of infecting humans and animals.

We assumed that humans acquire infections from human-specific STH, while animals contract infections from animal-specific STH.
We also assumed that individuals are born with no immunity to helminth parasites and enter the susceptible class with rates $b_h$ and $b_a$. They become exposed to helminth infection at the rates  $\lambda_h$ and $\lambda_a$, after coming into contact with the contaminated environment. Individuals in the exposed class are reinfected as they interact with the contaminated environment but remain latent for a duration of $\frac{1}{\rho_h}$ and  $\frac{1}{\rho_a}$, where $ \rho_h$ and $\rho_a$ are the progression rates from the exposed to the infectious class.

We assumed that an infected person increases the number of parasitic eggs in the soil at the rate $\varepsilon_h$ when they defecate outside the toilets or latrines. 
For simplicity, in our model, individuals in the infected class do not acquire new infections and will recover without treatment at a rate $\gamma_h$ with no period of immunity. Similarly, infected animals increase the number of parasitic eggs in the environment at the rate $\varepsilon_a$ when they roam without restrictions and recover temporally from infection at the rate $\gamma_a$.

Each population can become infected with parasitic eggs/larvae in the soil at rates $ \lambda_h =\frac{\beta_h M}{K + M }$ and $ \lambda_a =\frac{\beta_aM}{K + M }$, where $ \beta_h$ and $ \beta_a$ are the intake rates of eggs from contaminated food or larvae that have penetrated the skin, leading to infection. \textit{K} is the half-saturation constant of helminths within the environment. To simplify our model, we assumed that the number of parasitic eggs ingested is negligible compared to the total parasite population present in the environment. The natural mortality in each human class is denoted by $\mu_h$. 

We assumed that helminth-induced death does not occur in the human population, as STH causes morbidity but does not significantly contribute to mortality in humans. Animals in each class die naturally at a rate $\mu_a$, and infected ones contribute to the death count due to worm infection, at a rate $d_a$. Furthermore, we assumed that parasitic eggs/larvae die naturally at the rate $\mu_m$.
\subsection{Mathematical model}

We employed a compartmental model that categorizes human and animal populations into various states, with transitions between these states occurring at specific rates, as illustrated in Figure \ref{Figure:1}.
\begin{figure}[h]
\includegraphics[width=\textwidth]{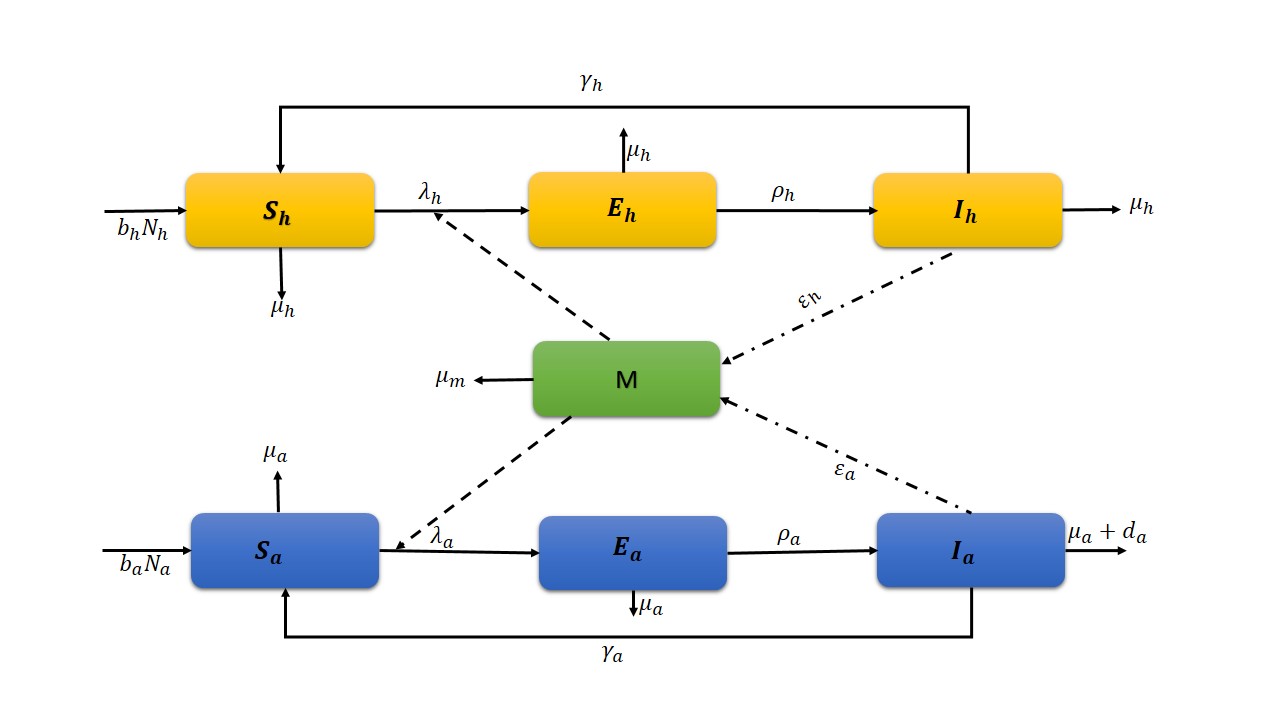}
\caption{\footnotesize{Compartmental diagram for the helminth model. The dashed lines indicate the interaction of the human and animal populations
with the contaminated environment that leads to the ingestion of parasitic eggs. The dotted dashed lines  indicate the release of eggs by individuals or animals
 infected with parasitic helminths, leading to environmental contamination.}}
\label{Figure:1}
\end{figure}


The parameters and their corresponding values used in model simulation are listed in Table \ref{tab:parameters}. The model accounts for seven nonlinear ordinary differential equations that are given by system \eqref{equations}.

\begin{equation}\label{equations}
\begin{cases}
\dfrac{dS_h}{dt}&=b_hN_h+{\gamma }_hI_h-\left({\mu }_h+{\lambda }_h\right)S_h\\~\\
\dfrac{dE_h}{dt}&={\lambda }_hS_h-\left({\mu }_h+{\rho }_h\right)E_h\\~\\
\dfrac{dI_h}{dt}&={\rho }_hE_h-\left({\mu }_h+\gamma_h\right)I_h\\~\\
\dfrac{dS_a}{dt}&=b_aN_a+{\gamma }_aI_a-\left({\mu }_a+{\lambda }_a\right)S_a\\~\\
\dfrac{dE_a}{dt}&={\lambda }_aS_a-\left({\mu }_a+{\rho }_a\right)E_a\\~\\
\dfrac{dI_a}{dt}&={\rho }_aE_a-\left({d_a+\mu }_a+\gamma_a\right)I_a\\~\\
\dfrac{dM}{dt}&={\varepsilon }_hI_h+{\varepsilon }_aI_a-{\mu }_mM
\end{cases}
\end{equation}

\begin{table}[H]
\centering
\caption{{\bf Parameter values (measured per day) for soil-transmitted helminth infection}}
\begin{threeparttable}
\begin{tabular}{p{6cm} p{1.5cm}  p{2cm}  p{4cm}}
\hline
\textbf{Parameter}&\textbf{Symbol} & \textbf{Value}& \textbf{Details} \\ \hline
Per capita birth rate of humans (Philippines) &\textbf{$b_h$} &$0.0000541$& \cite{bib48}\\
Ingestion rate of helminth parasite by humans &\textbf{$\beta_h$} &$0.02022$ & Fitted \\
Recovery rate without treatment of humans&\textbf{$\gamma_h$} &$0.02 $ & \cite{CDC} (based on hookworm study only) \\
Shedding rate of parasites into the environment by humans&\textbf{$\varepsilon_h$} &0.16462 & Fitted \\
Progression rate from exposed to infected class for humans&$\rho_h$ & $0.048$ & \cite{Baron1996}\\
Natural death rate of humans &\textbf{$\mu_h$} & $0.000038$ &\cite{bib49}\\
Ingestion rate of helminth parasites by animals&\textbf{$\beta_a$} &0.08456 & Fitted \\
Recovery rate of animals without treatment &\textbf{$\gamma_a$} & $0.01$ &Estimated to be lower than that of humans ($\gamma_h$)\\
Shedding rate of parasites into the environment by animals&\textbf{$\varepsilon_a$} &0.17661 & Fitted \\
Progression rate from exposed to infected class for animals&$\rho_a$ & 0.071 & \cite{Epe2009}\\
Per capita birth rate of animals &\textbf{$b_a$} & 0.0003 & \cite{New2024}\\
Natural death rate of animals &\textbf{$\mu_a$} & $0.0002$ & \cite{New2024}\\
Death rate caused by helminth on  animals&\textbf{$d_a$} & 0.0006 & Estimated to be thrice the value of $\mu_a$\\
Clearance rate of parasites from the environment &$\mu_m$ & $0.35546$ & Fitted \\ 
\hline
\end{tabular}
\end{threeparttable}
\label{tab:parameters}
\end{table}


\subsection{Solution of the mathematical model}
\vskip .1in
The total size of the human population at any time $ t $ is given by 
\begin{equation}\label{humans}
N_h(t)= S_h(t)+E_h(t)+I_h(t),
\end{equation}
while  that of the animal population is given by 
\begin{equation}\label{animals}
N_a(t)= S_a(t)+E_a(t)+I_a(t).
\end{equation}

The initial conditions are given by:
$S_h(0)>0 $, $E_h(0)\geq 0$, $I_h(0)\geq 0$, $S_a(0)> 0$, $E_a(0)\geq 0$, $I_a(0)\geq 0$, $M(0)\geq 0$. All parameters are assumed to be non-negative over the modeling time frame.

From \eqref{equations}, \eqref{humans} and \eqref{animals}, we have
\begin{align}\label{population}
\begin{cases}
\dfrac{dN_h}{dt}&=\left(b_h-{\mu _h}\right)N_h\\
\dfrac{dN_a}{dt}&=\left(b_a-{\mu _a}\right)N_a- d_aI_a.
\end{cases}
\end{align}

To analyze this model, we introduce the following change of variables.\\
$$ s_h = \dfrac{S_h}{N_h},\; e_h =\dfrac{E_h}{N_h},\;  i _h=\dfrac{I_h}{N_h},\; s_a = \dfrac{S_a}{N_a},\; e_a =\dfrac{E_a}{N_a},\; i _a=\dfrac{I_a}{N_a},\; \text{and}\; m = \dfrac{M}{K}$$\\

Consequently, \eqref{population} can be rewritten in terms of fractional variables as
\begin{align*}
\begin{cases}
\dfrac{dN_h}{dt}&=\left(b_h-{\mu _h}\right)N_h\\
\dfrac{dN_a}{dt}&=\left(b_a-{\mu _a}-d_a i_a\right)N_a.
\end{cases}
\end{align*}
Following the change of variables, model \eqref{equations} can be converted to its fractional form as follows:\\
\begin{equation}\label{fracequations}
\begin{cases}
\dfrac{ds_h}{dt}&=b_h+{\gamma }_h i_h-\left({b_h+\frac{\beta_h m}{1+m }}\right)s_h\\~\\
\dfrac{de_h}{dt}&=\frac{\beta_h m}{1+m}s_h-\left(b_h+{\rho }_h\right)e_h\\~\\
\dfrac{di_h}{dt}&={\rho }_he_h-\left(b_h+\gamma_h\right)i_h\\~\\
\dfrac{ds_a}{dt}&=b_a+{\gamma }_a i_a-\left({b_a+\frac{\beta_a m}{1+m }-d_ai_a}\right)s_a\\~\\
\dfrac{de_a}{dt}&=\frac{\beta_a m}{1+m}s_a-\left(b_a+{\rho }_a-d_ai_a\right)e_a\\~\\
\dfrac{di_a}{dt}&={\rho }_a e_a-\left(b_a+d_a+\gamma_a-d_ai_a\right)i_a\\~\\
\dfrac{dm}{dt}&={\varepsilon }_hi_h+{\varepsilon }_ai_a-\mu_mm
\end{cases}
\end{equation}
\vskip .1in
\noindent The solutions of system \eqref{fracequations} enter the positively invariant region given by
\begin{equation}
\Omega =\left \{(s_h,e_h,i_h,s_a,e_a,i_a,m) \in \mathbb R_+^{7}|s_h+e_h+i_h\leq1,s_a+e_a+i_a\leq  1,0\leq m \leq 1\right\}.
\end{equation}

We now show that the solution of system \eqref{fracequations} is non-negative for all time by proving the following theorem.
\begin{theorem}\label{thm:1}
With non-negative initial conditions, the solution set of system \eqref{fracequations}, 
$(s_h,e_h,i_h,s_a,e_a,i_a,m)$ remains non-negative for all time $t>0$.\\
\end{theorem}

\begin{proof}
From the first equation of \eqref{fracequations}, we have:
$$\frac{ds_h}{dt}=b_h+{\gamma }_h i_h-\left({b_h+\frac{\beta_h m}{1+m }}\right)s_h\geq  -\left({b_h+\frac{\beta_h m}{1+m }}\right)s_h.$$\\
Integrating with respect to $t$ with initial condition $s_h(0) $, we obtain
$$s_h(t)\geq s_h(0)exp\left\{-\int{ \left(b_h+\frac{\beta_h m}{1+m }\right)}\ dt\right\}> 0.$$
The same argument yields
\begin{align*}
e_h(t)\geq & e_h(0)exp\{-\left(b_h+\rho_h\right)t\}\geq 0,\\
i_h(t)\geq & i_h(0)exp\{- \left(b_h+\gamma_h\right)t\}\geq 0\\
s_a(t)\geq & s_a(0)exp\left\{-\int{ \left(b_a+\frac{\beta_a m}{1+m }-d_ai_a\right)}\ dt\right\}> 0,\\
e_a(t)\geq & e_a(0)exp\left\{-\int{ \left(b_a+\rho_a-d_ai_a\right)}\ dt\right\}\geq 0\\
i_a(t)\geq & i_a(0)exp\left\{-\int{ \left(b_a+d_a+\gamma_a-d_ai_a\right)}\ dt\right\}\geq 0\\
m(t)\geq & m(0)exp\left\{-\mu_mt\right\}\geq 0.\\
\end{align*}
Therefore, the solution set  $(s_h,e_h,i_h,s_a,e_a,i_a,m)$ of the model equation \eqref{fracequations} is non-negative for all $t>0$. 
\end{proof}

This shows that the model is well-posed and biologically meaningful, since the subpopulation cannot be negative. Next, we look at the the transmission potential of STH infection by computing the basic reproduction number $R_0$ of system \eqref{equations}, which is the average number of people that can be infected by one person.

\subsection{Existence of disease-free equilibruim}
The disease-free equilibrium (DFE) denoted as $E_0$ represents a state where infection is absent from the population. The DFE can be obtained by getting the roots of the model system. In our case,  setting \eqref{fracequations} to zero, i.e.,
$$\frac{ds_h}{dt}=\frac{de_h}{dt}=\frac{di_h}{dt}=\frac{ds_a}{dt}=\frac{de_a}{dt}=\frac{di_a}{dt}=\frac{dm}{dt}=0,$$ 
and solving the equations simultaneously, we obtain
$$ E_0=(1,0,0,1,0,0,0).$$

In this state, the entire population, both human and animal, is susceptible, as seen in the first and fourth terms of $E_0$. This equilibrium state is stable, i.e., the disease will not spread throughout the population if $R_0<1$. We next present a proof to support this claim and provide a computation of $R_0$.


\subsection{Basic reproduction number}
The basic reproduction number, $R_0$, quantifies the transmission potential of the infection. We use the Next Generation Matrix approach [32], to derive $R_0$. The Next Generation Matrix consists of two components, the new infection matrix $F$ and the transfer matrix $V$. The Jacobian matrices of $F$ and $V$ are derived at the disease-free equilibrium $E_0$. $R_0$ is then determined as the spectral radius of the matrix product $FV^{-1}$. Thus, 

$$R_0=\frac{\beta_a\varepsilon_a\rho_a\theta_h\omega_h+\beta_h\varepsilon_h\rho_h\theta_a\omega_a}{\mu_m\omega_h\theta_h\omega_a\theta_a},$$\\ which accounts for the contributions from both human and animal populations. The detailed derivation is presented in Appendix
 \ref{app:reproduction number}.

Understanding the basic reproduction number is essential for analyzing the stability of $E_0$, particularly when it undergoes perturbations. This analysis highlights the importance of assessing both local and global stability. Local stability examines the system’s response to small disturbances near the equilibrium, determining whether it has the ability to return to its original state. In contrast, global stability considers the system’s behavior across the entire state space, ensuring convergence to the equilibrium regardless of where the system started. By providing insights into whether a disease will persist or be eventually eradicated, stability analysis plays a critical role in crafting effective public health interventions. We perform these stability analyses of the DFE in the next subsection.

\subsection{Local and global stability of the DFE}
The local stability of the DFE is analyzed using the Jacobian matrix and eigenvalue analysis. The results indicate that local asymptotical stability exists for the disease-free equilibrium state, $E_0$ whenever $R_0 < 1$. 

\begin{theorem} \label{thm:local stability}
The disease-free equilibrium state, $E_0$, of the model system \eqref{fracequations} is locally asymptotically stable when $R_0<1$, and unstable for $R_0>1$.
\end{theorem}
The proof of Theorem \ref{thm:local stability} is presented in Appendix \ref{app:local stability}.

On the other hand, to determine the global stability of the disease-free equilibrium point within $ \Omega \subset R^7_+$, we employ the method proposed by Chavez \emph{et al.} \cite{bib34}. This approach divides the system into two subgroups; the infected and the uninfected subpopulation. We apply this method by denoting 
\eqref{fracequations} by
\begin{equation}
\begin{cases}
\frac{dX}{dt}=&F(X,Y).\\
\frac{dY}{dt}=&G(X,Y),\;G(X,0)=0.
\end{cases}
\end{equation}
Here, $X=(s_h,s_a)\in R^2_+$ denotes the uninfected subpopulation and $ Y=(e_h,i_h,e_a,i_a,m)\in R^5_+ $ denotes the infected and infectious subpopulations. The disease-free equilibrium can now be denoted by $E_0=(X^*,0)=(\mathbf{1},\mathbf{0})$, where $X^* = \mathbf{1}\in R^2_+$ and $\mathbf{0}\in R^5_+$. $E_0$ is globally asymptotically stable (GAS) if $R_0<1$ and the following two conditions $(D_1)$ and $(D_2)$ are  satisfied:
\begin{align*}
&D_1:\text{For}\; \frac{dX}{dt}= F(X,0),\;X^*\; \text{is} \;\text{GAS},\\
&D_2:G(X,Y)= AY-\widehat {G}(X,Y),\;\widehat {G}(X,Y)\geq0\;\text{for}\; (X,Y) \in \Omega, 
\end{align*}
where \begin{equation*}
A=D_YG(X^*,0)
\end{equation*}
is the Metzeler-matrix (the off-diagonal elements of $A$ are nonnegative)  and $\Omega$ is the region
where the model makes biological sense.

\begin{theorem} \label{thm:global stability}
The disease-free equilibrium point, $ E_0=(1,0,0,1,0,0,0)\in\Omega $ is GAS for the system \eqref{fracequations} provided that $R_0 < 1$, and conditions $(D_1)$ and $(D_2)$  both hold.
\end{theorem}
The proof of the preceding theorem can be found in Appendix \ref{app:global stability}.

Aside from the determination and analysis of the disease-free equilibrium, it is also important to determine whether the infection will persist in the population. To assess this, we establish the existence and find an expression of the endemic equilibrium, the state where the disease remains constantly present.

\subsection{Existence of endemic equilibruim}
The endemic equilibrium state refers to the equilibrium state where the disease persists in the population. In this state, all diseased states are considered to be positive, and as a result, $i_h$ and $i_a$ must be greater than zero for all other states to be positive.

If we let $\lambda_h^*=\dfrac{\beta_hm^*}{1+m^*}$, $\lambda_a^*=\dfrac{\beta_am^*}{1+m^*}$,  and $z^*_a=b_a-d_ai_a$  and equating system \eqref{fracequations} to
zero, the endemic equilibrium $(s_h^*,e_h^*,i_h^*,s_a^*,e_a^*,i_a^*,m^*)$ can be computed as\\

$ s_h^*=\dfrac{(b_h+\rho_h)(b_h+\gamma_h)}{(b_h)^2+(\Delta_{2h}+\lambda_h^*)b_h+\lambda_h^*\Delta_{2h}+\rho_h\Delta_{1h}}$\\

$ e_h^*=\dfrac{\lambda_h^*(b_h+\gamma_h)}{(b_h)^2+(\Delta_{2h}+\lambda_h^*)b_h+\lambda_h^*\Delta_{2h}+\rho_h\Delta_{1h}}$\\

$ i_h^*=\dfrac{\rho_h\lambda_h^*}{(b_h)^2+(\Delta_{2h}+\lambda_h^*)b_h+\lambda_h^*\Delta_{2h}+\rho_h\Delta_{1h}}$\\

$ s_a^*=\dfrac{b_a(\rho_a+z_a^*)(\Delta_{1a}+z_a^*)}{(z_a^*)^3+(\Delta_{2a}+\lambda_a^*)(z_a^*)^2+(\lambda_a^*\Delta_{2a}+\rho_a\Delta_{1a})z_a^*+\rho_a\lambda_a^*d_a}$\\

$ e_a^*=\dfrac{b_a\lambda_a^*(\Delta_{1a}+z_a^*)}{(z_a^*)^3+(\Delta_{2a}+\lambda_a^*)(z_a^*)^2+(\lambda_a^*\Delta_{2a}+\rho_a\Delta_{1a})z_a^*+\rho_a\lambda_a^*d_a}$\\

$ i_a^*=\dfrac{b_a\rho_a\lambda_a^*}{(z_a^*)^3+(\Delta_{2a}+\lambda_a^*)(z_a^*)^2+(\lambda_a^*\Delta_{2a}+\rho_a\Delta_{1a})z_a^*+\rho_a\lambda_a^*d_a}$

\begin{align*}
m^*&=\frac{b_a\rho_a\epsilon_a\lambda_a^*}{\mu_m\left((z_a^*)^3+(\Delta_{2a}+\lambda_a^*)(z_a^*)^2+(\lambda_a^*\Delta_{2a}+\rho_a\Delta_{1a})z_a^*+\rho_a\lambda_a^*d_a\right)}\\
 &+\frac{\rho_h\epsilon_h\lambda_h^*}{\mu_m\left({(b_h)^2+(\Delta_{2h}+\lambda_h^*)b_h+\lambda_h^*\Delta_{2h}+\rho_h\Delta_{1h}}\right)}.
\end{align*}

where $\Delta_{2a}=\Delta_{1a}+\rho_a , \Delta_{1a}=d_a+\gamma_a$,
$\Delta_{2h}=\gamma_h+\rho_h , \Delta_{1h}=\gamma_h.$

At the endemic equilibrium, the disease remains present in the population, and since all the parameters used are positive, it follows that all the computed state variables are greater than zero. Having computed the endemic equilibrium, it is essential to explore the stability and behavior of the endemic equilibrium. This brings us to the concept of bifurcation analysis, which allows us to investigate the critical points where the system transitions between different equilibrium states.

\subsection{Bifurcation analysis}
The bifurcation behavior of the model is examined using the center manifold theory proposed by Castillo-Chavez and Song [35]. In this theory, values $q$ and $p$ are computed; $q$ indicates the presence of a bifurcation and $p$ defines the nature of the bifurcation. The values of $q$ and $p$ are computed from the following expressions, details of which are explained in Appendix \ref{app:bifurcation}.

\begin{align*}
q=&\sum_{k,i=1}^7v_kw_i\frac{\partial^2f_k}{\partial x_i\partial \beta^*}(E_0,\beta^*).\\
p=&\sum_{k,i,j=1}^7v_kw_iw_j\frac{\partial^2f_k}{\partial x_i\partial x_j}(E_0,\beta^*),
\end{align*}
where $E_0$ is the disease-free equilibrium and $\beta^*$ is the bifurcation parameter. 

The analysis demonstrates that the coefficient $q$ is positive $(q>0)$  affirming its contribution to bifurcation dynamics.
The bifurcation type is determined by the sign of $p$.
If $p < 0$ then system \eqref{fracequations} will 
exhibit forward bifurcation; if $p > 0$, it undergoes backward bifurcation.
In forward bifurcation, the disease-free equilibrium remains stable for  $R_0<1$, while an endemic equilibrium emerges and remains stable only when $R_0>1$. Reducing $R_0 $ below 1 leads to the natural eradication of the infection. As a result, standard control measures such as vaccination and treatment, which lower $R_0 $ below 1, are effective in eliminating the disease.
Backward bifurcation, on the other hand, arises when a stable endemic equilibrium coexists with a stable disease-free equilibrium even when $R_0<1$. This implies that reducing $R_0 $ below 1 may not be sufficient to eradicate the disease, as the infection can persist. In such scenarios, additional interventions, such as increasing vaccination coverage, are necessary to achieve disease elimination [36].

Further details of the bifurcation analysis are available in Appendix \ref{app:bifurcation}.

\section{Parameter estimation and sensitivity analysis}\label{section:paramestimate}

\subsection{Parameter estimation}

Parameter estimation ensures that the model accurately represents the observed data. In this study, we performed parameter estimation using STH prevalence data to estimate our model's sensitive parameters. The process aimed to minimize a measure of error between the model output and the observed prevalence data in humans and animals.

The secondary data used for parameter estimation in this study was sourced from a survey done in two provinces in Mindanao, Philippines. The study \cite{bib47}  collected information on the zoonotic transmission of intestinal parasites in these locations. This data served as the benchmark for evaluating the accuracy of the model.
The parameters fitted in this study include:\\
\noindent
$\beta_h:$ Ingestion rate of helminth parasite by humans\\
$\varepsilon_h:$ Shedding rate of parasites into the environment by humans\\
$\beta_a:$ Ingestion rate of helminth parasite by animals\\
$\varepsilon_a:$ Shedding rate of parasites into the environment by animals\\
$\mu_m:$ Clearance rate of eggs from the environment

These parameters are critical as they influence the model's dynamics and ability to replicate real-world observations. We initiated the parameter estimation process with a grid search over the parameter space, employing the Latin Hypercube Sampling (LHS) method [37]. This ensures that the entire range of parameter values is efficiently explored by dividing the space into equally probable intervals. Once the parameters are sampled, a set of parameter combinations is evaluated, and the combination that results in the minimum error is selected. Numerical optimization is then applied to refine the parameter estimates further. This optimization process carried out using a Python intrinsic function, minimizes the error by adjusting the parameter values iteratively, improving the accuracy of the model's predictions.
The Python code used for the parameter estimation is publicly available on GitHub at https://github.com/imoro1984/STH.

\subsection {Sensitivity analysis}
Sensitivity analysis is essential for identifying the most effective strategies to mitigate the impact of soil-transmitted helminths (STH). It is often employed to evaluate the robustness of model predictions against variations in parameter values. Generally, the spread of the disease is directly linked to the basic reproduction number.  We therefore identified which parameters significantly influence the basic reproduction number, guiding us to prioritize them in identifying intervention strategies. 

Our method included both local and global computation of sensitivity indices for different parameters. While global sensitivity evaluates the influence of differences in all parameters across their entire range of possible values, local sensitivity looks at how small, incremental changes in a single parameter affect the model's output.

\subsubsection{Local sensitivity analysis}
In this section, we performed a sensitivity analysis on the basic reproduction number to evaluate the influence of each parameter on its value. Subsequently, we utilized Chitnis \emph{et al.’s} [38] method to compute the forward sensitivity index of the basic reproduction number to various parameters.  This index is defined as the ratio of the relative change in the variable to the relative change in the parameter. 
\newtheorem{definition}{Definition}
\begin{define}\cite{bib32} The normalized forward sensitivity index of a differentiable variable $u$ that depends  on parameter $\epsilon$  can be expressed as: $ r_\epsilon^u=\left(\dfrac{\partial u}{\partial\epsilon}\right)\cdot\left(\dfrac{\epsilon}{u}\right)$.
\end{define}

For example, the sensitivity index of the basic reproduction number $R_0$ to parameter  $\beta_h$ is
 given as $r_{\beta_h}^{R_0}=\left(\dfrac{\partial {R_0}}{\partial\beta_h}\right)\cdot\left(\dfrac{\beta_h}{R_0}\right)=0.05495$. The other indices were obtained similarly. Figure \ref{fig:localsen} shows the sensitivity index of the model parameters to $R_0$.

\begin{figure}[h]
\centering
    \includegraphics[width=0.8\textwidth]{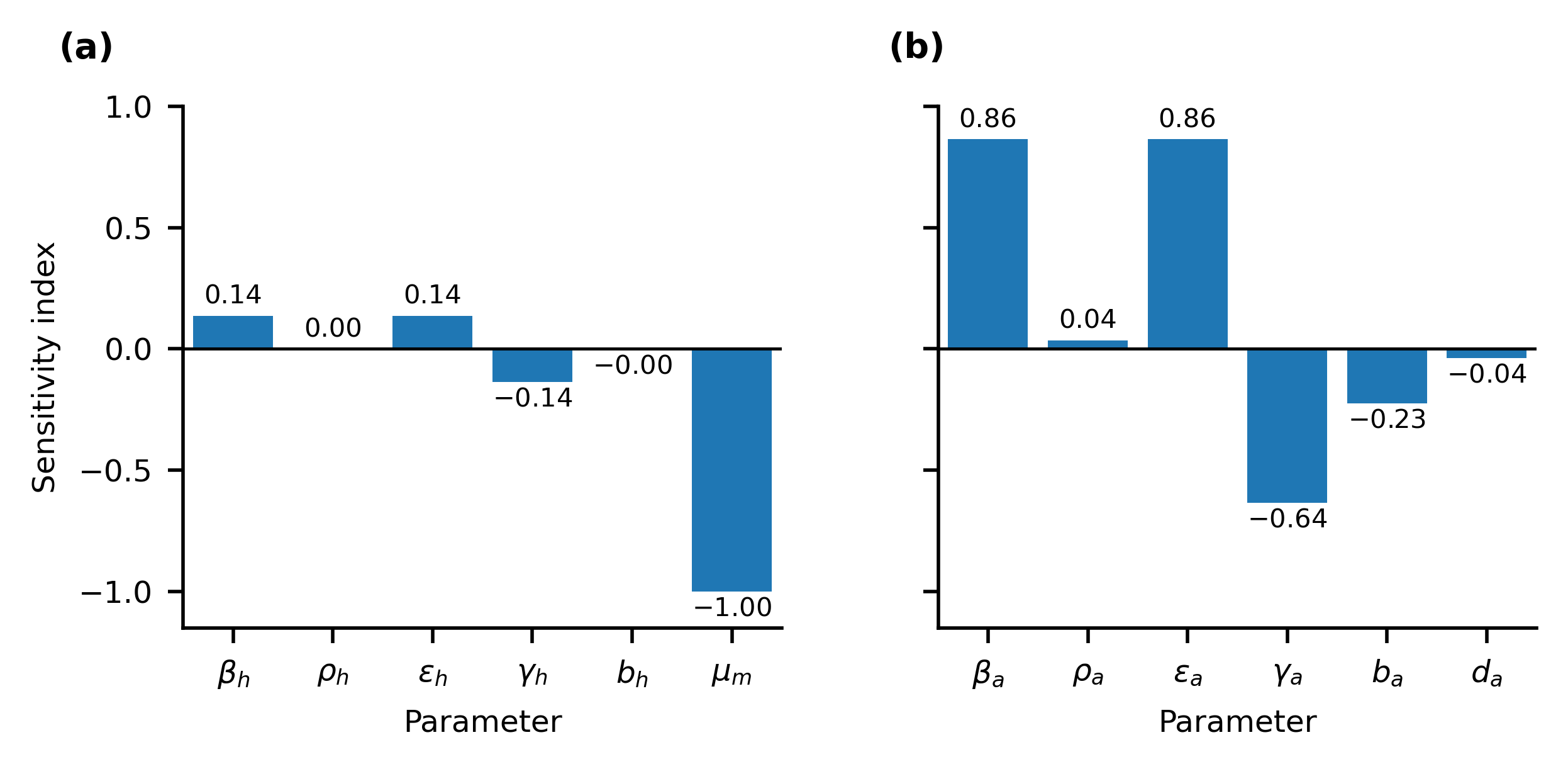}
\caption{Local sensitivity analysis of model parameters to $R_0$.  (a) Parameters associated with the human and parasitic egg population (b) Parameters associated with the animal population. Positive values (e.g., $\beta_a$, $\varepsilon_a$) suggest that increasing these parameters raises $R_0$, while negative values (e.g., $\mu_m$, $\gamma_a$, $b_a$) indicate that increasing these parameters reduces $R_0$ subsequently leading to a decrease in the infection.}
\label{fig:localsen}
\end{figure}

Parameters $\beta_h$, $\varepsilon_h$, $\beta_a$, and $\varepsilon_a$ exhibited positive sensitivity indices. This signifies that increasing the ingestion rates ($\beta_h$ \text{and} $\beta_a$), and the shedding rates ($\varepsilon_h$ \text{and} $\varepsilon_a$), while keeping other parameters constant will raise the disease's endemicity. 
A high sensitivity index indicates that the parameter significantly influences the disease's severity, while a low index reflects a minimal impact. This highlights that the ingestion rate of animals  ($\beta_a$) and their shedding rate ($\varepsilon_a $) have a greater effect on disease severity than the progression rate of humans ($\rho_h$).
 Conversely, the parameters,  $\mu_m$, $\gamma_h$, $\gamma_a$, $b_h$, $b_a$, and $d_a$  have negative sensitivity indices with the clearance rate of eggs from the environment ($\mu_m$) being the most influential and $b_h$ being the least. Increasing the values of $\mu_m$,  $\gamma_a$,  and $b_a$, while maintaining others, lowers the basic reproduction number, thus reducing the disease's endemicity. It is evident that parameters associated with the animal population significantly impact the disease dynamics, likely due to the high estimated prevalence rate of 56\% among animals, compared to an estimated prevalence of 19.6\% in humans. The significant role of animal-related parameters in the disease dynamics underscores the importance of considering the animal population in controlling STH infections in humans. This suggests that animals should be considered when developing control strategies for STHs. Additionally, promoting personal hygiene and maintaining environmental cleanliness are crucial for effective STH control.

\subsubsection{Global sensitivity analysis}

In this subsection, we conducted global sensitivity analyses on the 12 parameters in Table \ref{tab:parameters}. The goal was to assess each parameter's influence on the model's outcomes and quantify how input changes affect the results. This study focused on three key groups: infected humans, infected animals, and the number of parasitic eggs in the environment. To achieve this, we used a global sensitivity analysis method called partial rank correlation coefficient (PRCC) analysis, which is  well-known for its effectiveness and efficiency in sampling-based approaches. 

We first assigned a uniform distribution to each parameter, and applied Latin hypercube sampling (LHS)  [39] to generate input parameter values, which were utilized to conduct 5,000 simulations. The maximum and minimum values of the parameters were set at $ \pm90\%$  of the default values listed in Table 	\ref{tab:parameters}. 

PRCC values ranging from -1 to 1 were calculated at different time points using the MATLAB function partialcorr. Figure \ref{fig:senanalysis} (a)-(c) illustrates these PRCC values, where each bar represents a PRCC value at a specific instance. The bars depict how the sensitivity of the parameters evolves over time in relation to the studied population. A positive PRCC value indicates a linear relationship while a negative PRCC value means an inverse relationship. Also, a high absolute PRCC value suggests a strong correlation of the parameter with the model outcome, meaning a small change in that parameter could significantly impact the model's dynamics. 

\begin{figure}[h]
\includegraphics[width=0.8\textwidth]{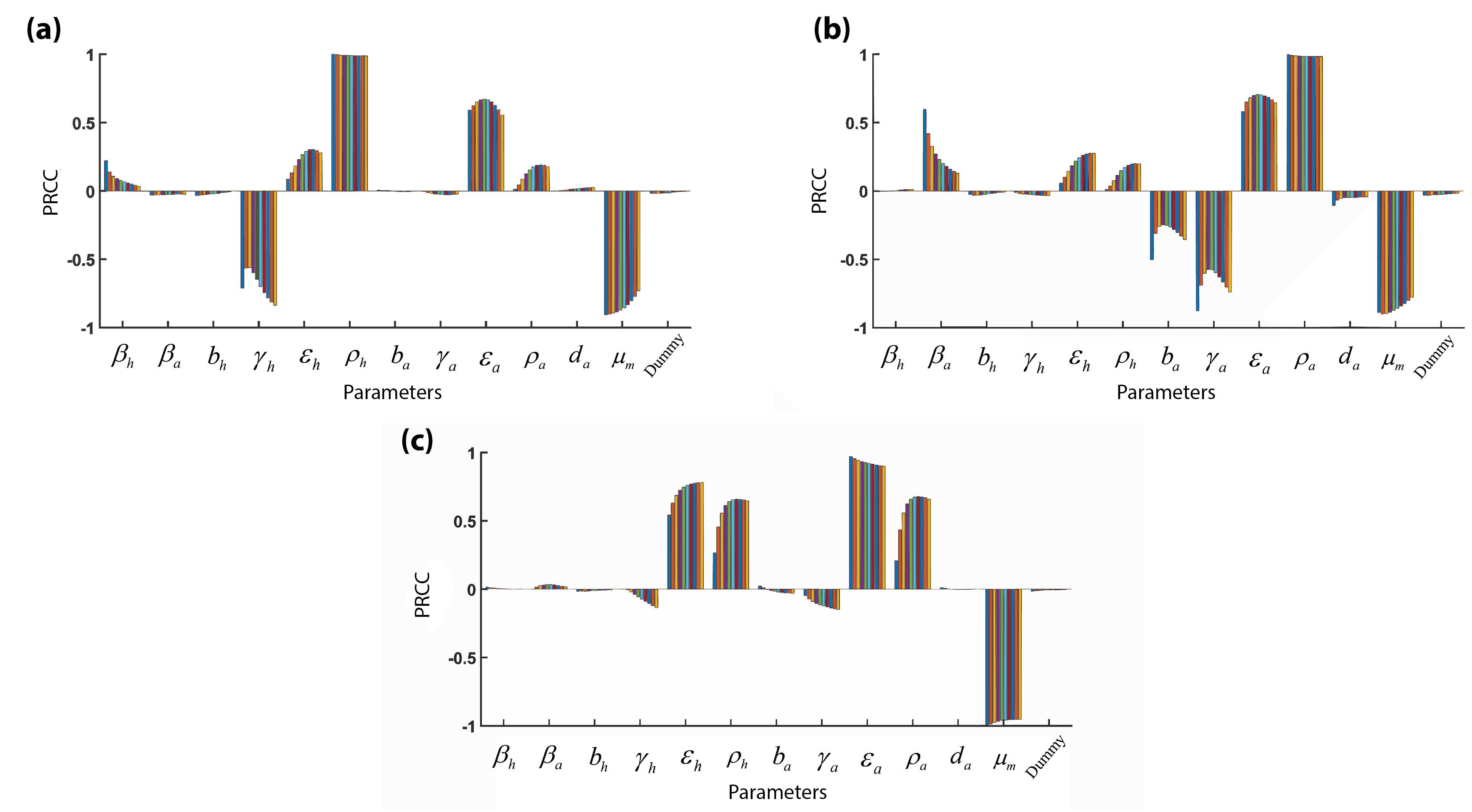}
\caption{PRCC values depicting the sensitivity of the model's output over time for (a) Infected humans (b) Infected animals and  (c) Parasitic egg population to the model parameters. Higher PRCC values signify a stronger correlation between a parameter and the population analyzed. The ingestion rates $\beta_a,\beta_h$, shedding rates $\varepsilon_a,\varepsilon_h$, and  progression rates $ \rho_a,\rho_h$ exhibit positive PRCC values whilst the recovery rates $\gamma_a, \gamma_h$ and clearance rate $\mu_m$ show negative PRCC values. }
\label{fig:senanalysis}
\end{figure}

PRCC values of the parameters to the infected human population are illustrated in Figure  \ref{fig:senanalysis} (a).  Parameters such as $\beta_h,\rho_h,\varepsilon_h,\rho_a,$ and $\varepsilon_a$ are found to have positive PRCC values indicating that an increase in the values of these parameters will increase the infected human population. $\rho_h$ and $\varepsilon_a$ shows a strong correlation while $ \beta_h$ is the least correlated to the infected human population. Conversely, parameters $\gamma_h$ and $\mu_m$ have negative PRCC values indicating that an increase in these values will result in a decrease in the infected human population. 

Parameters  $\beta_a$, $\varepsilon_a$ and $\rho_a$ were found to have high positive PRCC values, suggesting that an increase in these parameters will lead to a rise in the number of infected animals. On the other hand, parameters $b_a,\gamma_a$, and $\mu_m$  have negative PRCC values indicating that an increase in their values will reduce the number of infected animals, as shown in Figure \ref{fig:senanalysis} (b).

Meanwhile, parameters $\rho_h,\rho_a,\varepsilon_h$, and $\varepsilon_a$   have high positive PRCC values, to the parasitic egg population while parameter $\mu_m$ have high negative PRCC value. The corresponding values are presented in Figure \ref{fig:senanalysis} (c).  

It was noted that the PRCC indices of $\beta_h$ and $\beta_a$ exhibit a decreasing trend, and by the time equilibrium is reached, their effect becomes insignificant.
A dummy parameter was included to ensure the robustness of the analysis. Since it exhibited low sensitivity, our conclusions are deemed reliable. The global sensitivity analysis confirmed the sign direction (positive or negative) of the effect of parameters on the model as identified by local sensitivity analysis. However, the impact of parameters on the model is more pronounced in the global sensitivity analysis than in the local one. This is because global sensitivity analysis assesses the influence of parameters across the entire parameter space, considering interactions between parameters whilst local sensitivity analysis typically examines changes around a specific point in the parameter space, often the default or estimated parameter values.


\section{Numerical simulation and discussion}\label{section:simulation}
For our numerical simulation, we utilized the parameter values provided in Table \ref{tab:parameters}. We observed the population distribution profiles for humans, animals, and parasitic eggs in the contaminated environment. We made the following observations. There was a significant decrease in the susceptible population during the first $11$ days, which led to an increase in the number of exposed individuals and a rise in the number of infected individuals. Around the ninth day, the number of exposed individuals started to decline after reaching a peak of 0.6, while the infected human population continued to increase. The susceptible individuals eventually stabilized. A similar trend was observed in the animal population, with a peak of 0.5 also occurring on the ninth day. These dynamics are illustrated in Figure \ref{fig:profile_distributionsh} (a)-(c) and Figure \ref{fig:profile_distributionsa} (a)-(c). 

Certain key parameters (ingestion, shedding, and clearance rates) were modified while keeping all others at their baseline values, allowing for an analysis of their impact on infection dynamics across the different populations, leading to the following insights. Increasing $\varepsilon_a$,  and $\varepsilon_h$ causes a significant increase in the infected individuals $i_h$ and $i_a$  and the parasite population as demonstrated in Figures \ref{fig:profile_distributions} (a)-(c). On the other hand, increasing $\mu_m$  leads to a decrease in the three populations. This is illustrated in Figure \ref{fig:profile_distributions} (a) - (c).
It can also be observed that  $\varepsilon_a$, and $\varepsilon_h$ influence both the infected human and animal populations. However, the parameters associated with the animal population have a more pronounced effect on the human population compared to the impact of human parameters on the infected animal population. These trends are illustrated in Figure \ref{fig:profile_distributions} (a) and (b).
 An increase in $\beta_a$, and $\beta_h$ has little effect on the three populations, as illustrated in Figure \ref{fig:profile_distributions} (a)-(c).

\begin{figure}[H]
\centering
        \includegraphics[width=0.8\textwidth]{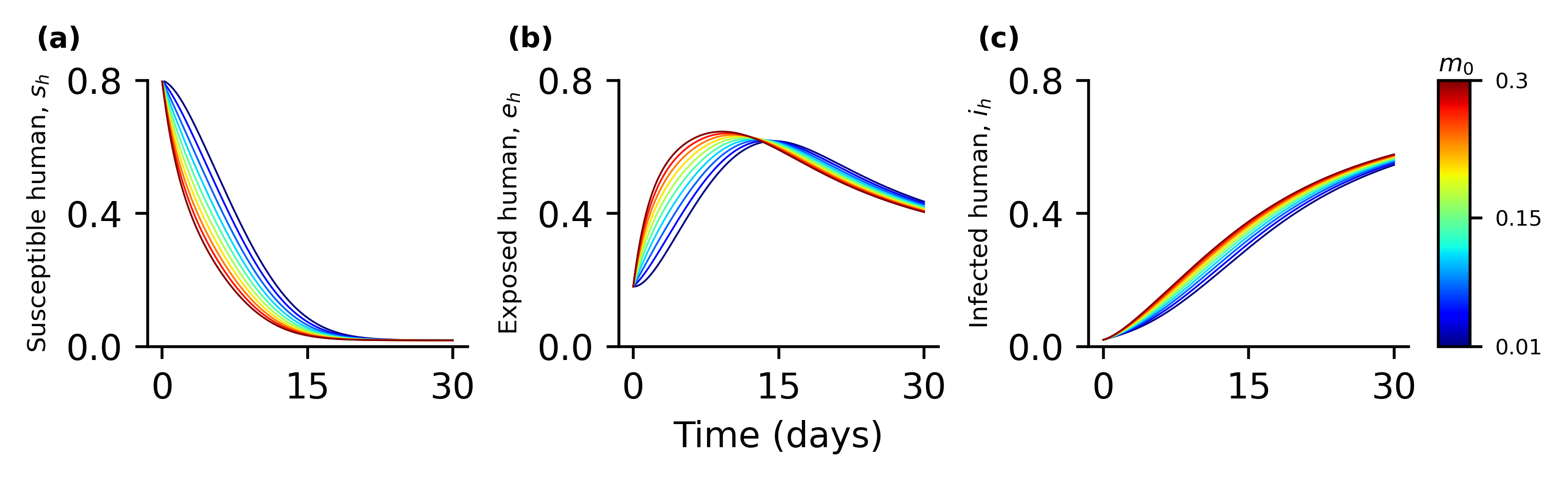}
    
\caption{Profile distribution of the human population  (a) Susceptible, (b) Exposed, and (c) Infected for varying initial parasitic egg populations in the contaminated environment $m_0$. A rapid decline in the susceptible population and an increase in the exposed and infected populations are observed. As $m_0$ increases, the decline in the susceptible population occurs over a longer period, and fewer individuals are exposed.}
\label{fig:profile_distributionsh}
\end{figure}

\begin{figure}[H]
\centering
        \includegraphics[width=0.8\textwidth]{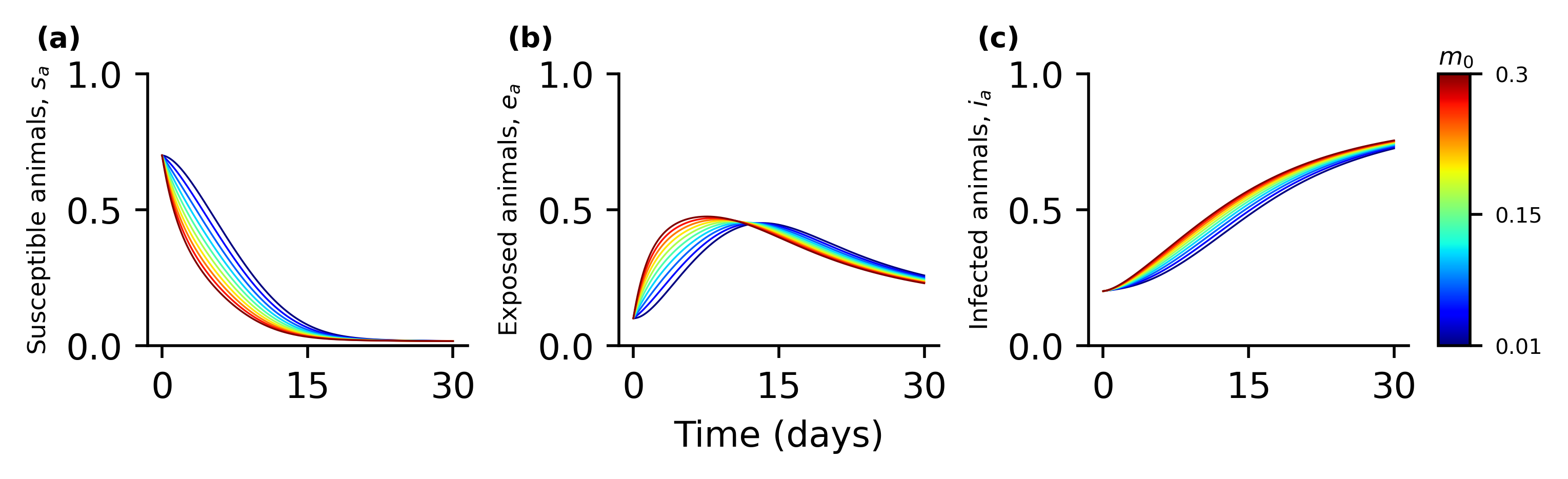}
        
\caption{ Profile distribution of the animal population  (a) Susceptible, (b) Exposed, and (c) Infected for varying initial parasitic egg populations in the contaminated environment $m_0$.
The susceptible population declines rapidly, while the exposed and infected populations increase. As $m_0$, decreases, the decline in the susceptible population takes longer and fewer individuals are exposed. A similar trend is observed in the infected animal population as $m_0$ decreases.}
\label{fig:profile_distributionsa}
\end{figure}

\begin{figure}[H]
\centering
        \includegraphics[width=0.9\textwidth]{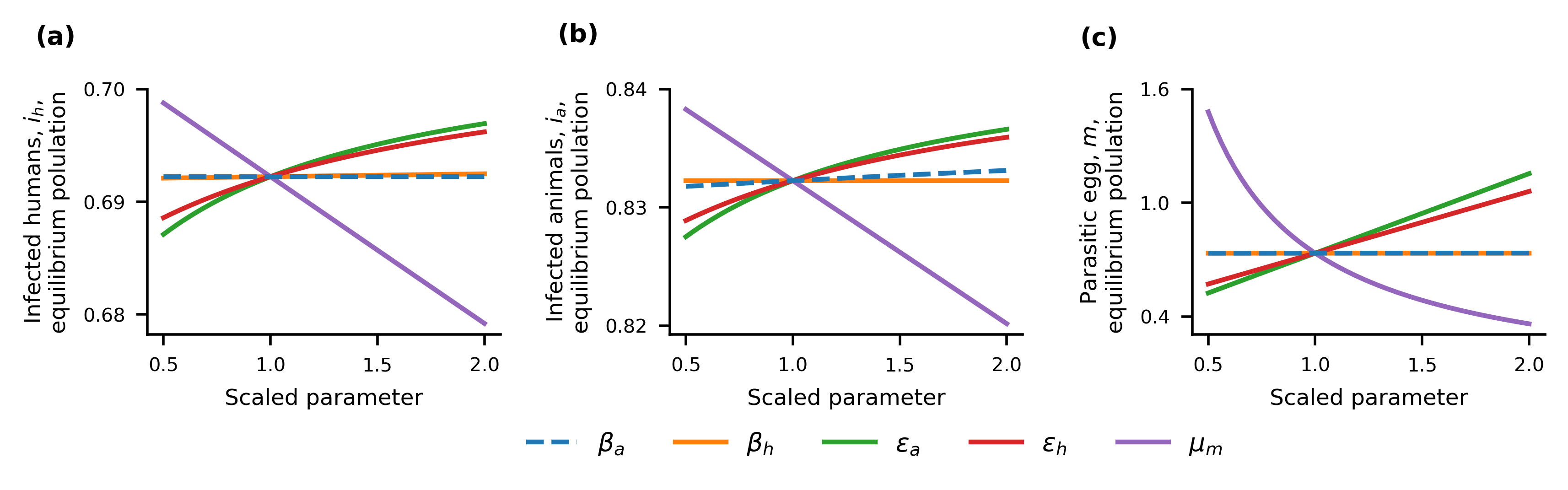}
        
\caption{Effect of varying  $\beta_a$,  $\beta_h$,  $\varepsilon_a$,  $\varepsilon_h$, and  $\mu_m$ on the equilibrium populations of (a) Infected humans, (b) Infected animals, and (c) Parasitic egg population. The parameters are scaled based on the estimated values in Table \ref{tab:parameters}, where a value of 1 represents the original estimate. Notably, in each curve, only the parameter of interest was altered, while all others remained at their baseline values. An increase in $\varepsilon_a$, and $\varepsilon_h$ leads to a higher number of infections, while a rise in $\mu_m$ results in a reduction. }
\label{fig:profile_distributions}
\end{figure}


We examined the transmission dynamics of STH using a susceptible-exposed-infected-susceptible (SEIS) compartment model. Our model does not include a recovery class since recovery from STH does not result in permanent immunity. The model encompasses three populations: humans, animals, and parasitic egg in a contaminated environment. Incorporating the animal population into the model is crucial because STH affects both humans and animals [16]. It is important to note that our data did not involve molecular techniques to determine if humans were infected with animal STH species or if animals were infected with human STH species.

Figures \ref{fig:profile_distributionsh} (a), and  \ref{fig:profile_distributionsa} (a) reveal a trend with an initial rapid decline of the susceptible population followed by a saturation point where the susceptible population no longer decreases significantly, occurring within the first $11$ days. This trend may result from a high ingestion rate, causing a significant portion of the susceptible population to become exposed quickly. The high initial parasite population in the model, leading to increased exposure, could also contribute to this trend. The 19\% prevalence of STH in the human population may further explain this pattern.
After the rapid decline, the susceptible population stabilizes, which could indicate that most of the susceptible individuals have either become exposed or infected or that the parasite density in the environment has decreased due to a high clearance rate, resulting in lower infection levels. The specific trend observed could also be attributed to the parameters used in our model, such as ingestion rate, shedding rate, and initial conditions.

From Figures \ref{fig:profile_distributionsh} (a)-(c), and  \ref{fig:profile_distributionsa} (a)-(c) we notice that as the initial parasite population decreases, there is a shift towards the right, indicating it takes longer days for most of the susceptible population to become infected. This suggests that a high initial population may explain the trend seen in Figures  \ref{fig:profile_distributionsh} (a), and  \ref{fig:profile_distributionsa} (a)

Figures \ref{fig:profile_distributions} (a)-(c) suggest that increasing the clearance rate and decreasing the shedding rate of parasitic eggs into the environment decreases the parasite population, consequently reducing the number of infected humans and animals. This underscores the importance of environmental sanitation and interventions that reduce environmental contamination in controlling STH infections. It also suggests that environmental factors are crucial in maintaining or reducing parasite burdens in populations. The pattern observed in  Figures \ref{fig:profile_distributions} (a)-(c) could imply the potential for cross-species transmission of STH between humans and animals, as a shared environment could facilitate this transmission. This highlights the need for further research to confirm cross-species transmission of the parasite, potentially involving molecular analysis to verify zoonotic transmission, as STH has zoonotic potential.
We also observe in Figure \ref{fig:profile_distributions} (c) that an increase in  $\beta_a$, and $\beta_h$ does not have a significant impact on the the parasite population $m$. This is due to the assumption that the number of parasitic eggs ingested is minimal in comparison to the overall parasitic egg population in the environment. A similar trend is observed in Figure \ref{fig:senanalysis} (c).

\section{Conclusion and Outlook}\label{section:conclusion}

In conclusion, the analysis of the STH model revealed several critical insights into the dynamics of infection spread and control. The model demonstrated that increasing ingestion rates in both human and animal populations has little impact on the number of infected individuals. This trend is evident in the PRCC analysis, where the number of infected individuals decreases over time. By the time equilibrium is reached, the PRCC indices are very small, indicating minimal influence. Furthermore, ingestion rates do not affect the parasitic egg load in the environment, as it is assumed that the number of ingested eggs is insignificant compared to the overall parasite population in the environment.  The model also clearly defines the importance of Environmental Control. The rate at which parasites are shed into the environment and the rate at which they are cleared play a crucial role in managing STH infections. This suggests that environmental interventions, such as improved sanitation and proper disposal of fecal waste of humans and animals, are vital for reducing infection rates in both human and animal populations. Cross-species transmission between animals and humans significantly impacts on the infection dynamics, especially in humans. This emphasizes the need for a One Health approach that
 considers the interconnectedness of human, animal, and environmental health in managing STH infections. The findings suggest that interventions should prioritize promoting personal hygiene, enhancing environmental sanitation, and controlling
 cross-species transmission. Overall, this model highlights the complexity of STH infection dynamics and the need for multi-faceted approaches to effectively manage and reduce the burden of these infections in both human and animal populations. 

In our succeeding study, we will focus on several strategies to mitigate the spread of STH infections. Model simulations indicate that both human and animal populations need protection against infections. Therefore, we recommend different control strategies which will be subject to further studies such as cost-effectiveness: 
\begin{itemize}
\item Anthelminthic Medications: Mass drug administration (MDA) to high-risk groups should be extended to animals under a One Health Approach. Including animals in MDA programs can help eliminate STH from both humans and animals by reducing the shedding of parasitic eggs into the environment. Previous studies suggest that medication alone may not eliminate STH  infections, necessitating alternative strategies \cite{bib40}.
\item  Health Education: Increasing public awareness of disease causes, transmission, and impacts is crucial for improving public health. STH infections are closely linked to poverty and low education levels \cite{bib41}, and educating the public can improve understanding and promote preventive measures. Given the evidence that animals contribute to the overall burden of intestinal helminth infections in humans, it is vital to educate the entire population, especially pet owners, on proper animal care during health campaigns. This can significantly reduce the ingestion rate of parasites.
\item  Water, Sanitation, and Hygiene(WASH): Enhancing WASH practices is crucial. Actions like wearing shoes and washing hands with soap before and after meals can reduce infection risks. Although sanitation has been the main priority, ensuring clean water consumption and promoting good hygiene can also help decrease STH prevalence. WASH initiatives should consider both humans and animals, which is why Alexandra et al. suggest that the "A" in WASH should stand for "Animal" \cite{bib42}.
\item  Use of Saprophytic Fungi: Applying saprophytic fungi to heavily contaminated soil can address STH infections by disrupting the parasite life cycle in the soil \cite{bib43}. This approach is crucial because STH eggs can survive in the environment for extended periods, up to 10 years under favorable conditions \cite{bib44}. This method can prevent parasitic eggs from developing into infectious stages, helping eliminate parasites from the environment and increasing the clearance rate from contaminated areas.   
\end{itemize}


\subsection*{Availability of data and materials}
 All data utilized in this analysis can be obtained by contacting the corresponding author.
\subsection*{Author Contributions} 
All authors were involved in the conception and design of the study, and the manuscript's writing. All authors have reviewed and approved the final version of the manuscript.
 \subsection*{Declaration of Competing Interest}
The authors confirm that they have no financial interests or personal relationships that could have influenced the research reported in this paper.

\newpage
\begin{appendices}
\section{Basic reproduction number}\label{app:reproduction number}
A key concept in epidemiology is the basic reproductive number, commonly denoted by $R_0$. Usually,  $R_0$ is defined as the average number of new cases reproduced in a wholly susceptible population when an infective individual is introduced into the population \cite{bib32}. It is the quantity that governs the transmission potentials of an infectious disease. 

Additionally, it helps in understanding the underlying dynamics and making predictions about the behavior of infectious diseases. When $R_0<1$, it means that the disease will die out in the population, while $R_0>1$,  implies the persistence of the disease in the population. $R_0$ is computed using the Next Generation Matrix approach by Van den Driessche and Watmough [32] at the disease-free equilibrium state, $E_0$. 

The Next Generation Matrix comprises two parts: The matrix of new infection $F$ and the transfer matrix
 $ V$. $\mathscr{ F}_i$ is the rate at which previously uninfected individuals enter compartment $\textit{i}$ whilst $\mathscr{ V}_i$ is the rate of transfer of individuals out of compartment $\textit{i}$ minus the rate of transfer into compartment $\textit{j}$.

$F$ has elements $ F_{ij}$ which is expressed as  $F_{ij}=\dfrac{\partial \mathscr{ F}_i(E_0)}{\partial x_j}$ and ${V}$  has elements $ V_{ij}$ which can also be expressed as $V_{ij}=\dfrac{\partial \mathscr{V}_i(E_0)}{\partial x_j}$.
 $R_0$ is the dominant eigenvalue or the spectral radius of the matrix, $FV^{-1}$ with the associated matrix $F$ and matrix $V$ as the Jacobian matrices of $F_{ij}$ and $V_{ij}$ respectively at the disease-free equilibrium state [32]. 

We derive $R_0$  using the Next Generation Matrix approach on the system (6) which is composed of the infected compartments. 
\begin{equation}
\begin{cases}
\dfrac{de_h}{dt}&=\frac{\beta_h m}{1+m}s_h-\left(b_h+{\rho }_h\right)e_h\\~\\
\dfrac{di_h}{dt}&={\rho }_he_h-\left(b_h+\gamma_h\right)i_h\\~\\
\dfrac{de_a}{dt}&=\frac{\beta_a m}{1+m}s_a-\left(b_a+{\rho }_a-d_ai_a\right)e_a\\~\\
\dfrac{di_a}{dt}&={\rho }_a e_a-\left(b_a+d_a+\gamma_a-d_ai_a\right)i_a\\~\\
\dfrac{dm}{dt}&={\varepsilon }_hi_h+{\varepsilon }_ai_a-\mu_mm
\end{cases}
\end{equation}
where $x_{i.j}=(x_1,x_2,x_3,x_4,x_5)=(e_h,i_h,e_a,i_a,m)$\\
From system (6):
$$ \mathscr{ F}=\left(\frac{\beta_h m}{1+ m}s_h,0,\frac{\beta_a m}{1+m}s_a,0,0\right)^T$$

$$\mathscr{ V}= \begin{bmatrix} 
(b_h+\rho_h)e_h\\
\rho_he_h-(b_h+\gamma_h)i_h\\
(b_a+\rho_a-d_a i_a)e_a\\
\rho_ae_a-\left(b_a+d_a+\gamma_a-d_ai_a\right)i_a\\
\varepsilon_hi_h+\varepsilon_ai_a-\mu_mm
\end{bmatrix}$$\\
The Jacobian matrices of $\mathscr{ F}$ and $\mathscr{ V}$ at the disease-free equilibrium $E_0$ are;
$$F=\begin{bmatrix}
0&0&0&0&\beta_h\\
0&0&0&0&0\\
0&0&0&0&\beta_a\\
0&0&0&0&0\\
0&0&0&0&0\\
\end{bmatrix}$$
$$V=\begin{bmatrix}
\theta_h&0&0&0&0\\
\rho_h&-\omega_h&0&0&0\\
0&0&\theta_a&0&0\\
0&0&\rho_a&-\omega_a&0\\
0&\varepsilon_h&0&\varepsilon_a&-\mu_m\\
\end{bmatrix}$$\\
where $\theta_h=b_h+\rho_h,\theta_a=b_a+\rho_a,\omega_h = b_h+\gamma_h,\omega_a=b_a+d_a+\gamma_a$
$$V^{-1}=\begin{bmatrix}
\frac{1}{\theta_h}&0&0&0&0\\
\frac{\rho_h}{\theta_h\omega_h}&-\frac{1}{\omega_h}&0&0&0\\
0&0&\frac{1}{\theta_a}&0&0\\
0&0&\frac{\rho_a}{\theta_a\omega_a}&-\frac{1}{\omega_a}&0\\
\frac{\varepsilon_h\rho_h}{\theta_h\omega_h\mu_m}&-\frac{\varepsilon_h}{\omega_h\mu_m}&\frac{\varepsilon_a\rho_a}{\theta_a\omega_a\mu_m}&-\frac{\varepsilon_a}{\omega_a\mu_m}&-\frac{1}{\mu_m}\\
\end{bmatrix}$$\\
$$FV^{-1}=\begin{bmatrix}
\frac{\varepsilon_h\beta_h\rho_h}{\theta_h\omega_h\mu_m}&-\frac{\varepsilon_h\beta_h}{\omega_h\mu_m}&\frac{\beta_h\varepsilon_a\rho_a}{\theta_a\omega_a\mu_m}&-\frac{\beta_h\varepsilon_a}{\omega_a\mu_m}&-\frac{\beta_h}{\mu_m}\\
0&0&0&0&0\\
\frac{\beta_a\varepsilon_h\rho_h}{\theta_h\omega_h\mu_m}&-\frac{\beta_a\varepsilon_h}{\omega_h\mu_m}&\frac{\beta_a\varepsilon_a\rho_a}{\theta_a\omega_a\mu_m}&-\frac{\beta_a\varepsilon_a}{\omega_a\mu_m}&-\frac{\beta_a}{\mu_m}\\
0&0&0&0&0\\
0&0&0&0&0\\
\end{bmatrix}$$\\
The basic reproduction number $R_0$ is the spectral radius of the next-generation matrix $ FV^{-1}$. Thus,
$$R_0=\frac{\beta_a\varepsilon_a\rho_a\theta_h\omega_h+\beta_h\varepsilon_h\rho_h\theta_a\omega_a}{\mu_m\omega_h\theta_h\omega_a\theta_a}$$\\
$$R_0=R_{0a}+R_{0h}$$\\
where\\$ R_{0a}=\frac{\beta_a\varepsilon_a\rho_a}{\mu_m\omega_a\theta_a}$\\
$R_{0h}=\frac{\beta_h\varepsilon_h\rho_h}{\mu_m\omega_h\theta_h}$
\section{Proof of Theorem \ref{thm:local stability}}\label{app:local stability}
We compute the Jacobian matrix of our model  at the disease-free equilibrium state $E_0$ as follows:
$$J_{E_0}=\scriptsize\begin{bmatrix}
-b_h&0&\gamma_h&0&0&0&-\beta_h\\
0&-(b_h+\rho_h)&0&0&0&0&\beta_h\\
0&\rho_h&-(b_h+\gamma_h)&0&0&0&0\\
0&0&0&-b_a&0&d_a+\gamma_a&-\beta_a\\
0&0&0&0&-(b_a+\rho_a)&0&\beta_a\\
0&0&0&0&\rho_a&-(b_a+d_a+\gamma_a)&0\\
0&0&\varepsilon_h&0&0&\varepsilon_a&-\mu_m\\
\end{bmatrix}$$\\
$$J_{E_0}=\scriptsize\begin{bmatrix}
-b_h&0&\gamma_h&0&0&0&-\beta_h\\
0&-\theta_h&0&0&0&0&\beta_h\\
0&\rho_h&-\omega_h&0&0&0&0\\
0&0&0&-b_a&0&\psi_a&-\beta_a\\
0&0&0&0&-\theta_a&0&\beta_a\\
0&0&0&0&\rho_a&-\omega_a&0\\
0&0&\varepsilon_h&0&0&\varepsilon_a&-\mu_m\\
\end{bmatrix}$$
where $\psi_a=d_a+\gamma_a$.

The first and second eigenvalues of the Jacobian matrix, $J_{E0}$ are $\lambda_1=-b_h,\lambda_2=-b_a$ which are strictly negative. The remaining five can be obtained by considering the submatrix:\\
$$\scriptsize\begin{bmatrix}
-\theta_h&0&0&0&\beta_h\\
\rho_h&-\omega_h&0&0&0\\
0&0&-\theta_a&0&\beta_a\\
0&0&\rho_a&-\omega_a&0\\
0&\varepsilon_h&0&\varepsilon_a&-\mu_m\\
\end{bmatrix}$$\\
The characteristic equation of the submatrix is thus given by\\
$\lambda^5+A_4\lambda^4+A_3\lambda^3+A_2\lambda^2+A_1\lambda+A_0=0$\\
where

$A_4=\mu_m+\omega_a+\omega_h+\theta_a+\theta_h$\\

$A_3=\mu_m(\omega_a+\theta_a)+\mu_m(\omega_h+\theta_h)+(\omega_a+\theta_a)(\omega_h+\theta_h)+\omega_a\theta_a+\omega_h\theta_h$\\

$A_2=\mu_m(\omega_h+\theta_h)(\omega_a+\theta_a)+\mu_m\omega_a\theta_a(1-R_{0a})+\mu_m\omega_h\theta_h(1-R_{0h})+\omega_a\theta_a(\omega_h+\theta_h)+\omega_h\theta_h(\omega_a+\theta_a)$\\

$A_1=\mu_m\omega_a\omega_h\left[(\theta_a(1-R_{0a})+\theta_h(1-R_{0h})\right]+\mu_m\theta_a\theta_h\left[(\omega_a(1-R_{0a})+\omega_h(1-R_{0h})\right]+\omega_a\omega_h\theta_a\theta_h$\\

$A_0=\mu_m\omega_a\omega_h\theta_a\theta_h\left(1-R_0\right)$\\

Applying the Routh-Hurwitz criteria \cite{bib33}, the roots of the characteristic equation of the submatrix have negative
real parts if the following inequalities are satisfied: $(i) A_4>0, A_3>0, A_2>0, A_1>0, A_0>0$.
$(ii) A_4A_3-A_2>0$ $(iii)A_2(A_4A_3-A_2)-A_1A_4^2>0$ $(iv) 2A_4A_1-A_3A_4^2+A_3A_2-A_0>0$ The inequalities $ ii,\; iii \; and \; iv$ are satisfied provided $A_i>0 (i=0,1,2,3,4)$  whenever $R_0<1$. Hence, according to the Routh-Hurwitz criteria, the submatrix has negative real parts whenever $R_0<1$.\\
Therefore, a local asymptotical stability exists for the disease-free equilibrium state, $E_0$  whenever $R_0<1$.


\section{Proof of Theorem \ref{thm:global stability}}\label{app:global stability}

From the system of equation \eqref{fracequations}, we have

$$F(X,Y)=\begin{bmatrix}
b_h+{\gamma }_h i_h-\left({b_h+\frac{\beta_h m}{1+m }}\right)s_h\\~\\
b_a+{\gamma }_a i_a-\left({b_a+\frac{\beta_a m}{1+m }-d_ai_a}\right)s_a
\end{bmatrix}$$

$$G(X,Y)=\begin{bmatrix}
\frac{\beta_h m}{1+m}s_h-\left(b_h+{\rho }_h\right)e_h\\~\\
{\rho }_he_h-\left(b_h+\gamma_h\right)i_h\\~\\
\frac{\beta_a m}{1+m}s_a-\left(b_a+{\rho }_a-d_ai_a\right)e_a\\~\\
{\rho }_a e_a-\left(b_a+d_a+\gamma_a-d_ai_a\right)i_a\\~\\
{\varepsilon }_hi_h+{\varepsilon }_ai_a-{\mu_m }m
\end{bmatrix}$$\\

\begin{equation}
\left.\frac{dX}{dt}\right|_{Y=0}=\begin{pmatrix}
b_h-b_hs_h \\ b_a-b_as_a\\

\end{pmatrix}
\end{equation}

$$\frac{ds_h}{dt}=b_h-b_hs_h$$
$$s_h(t)=1+\left(s_h(0)-1\right)e^{-b_ht}$$
As $t  \to \infty, s_h(t) \to 1$\\ also,
$$\frac{ds_a}{dt}=b_a-b_as_a$$
$$s_a(t)=1+\left(s_a(0)-1\right)e^{-b_at}$$
As $t  \to \infty, s_a(t) \to 1$
implying global convergence of solution of $(8)$  in $\Omega$. Thus, condition  $(D_1)$ is satisfied.\\

$$A=\scriptsize\begin{bmatrix}
-(b_h+\rho_h)&0&0&0&\beta_h\\
\rho_h&-({b_h+\gamma_h})&0&0&0\\
0&0&-(b_a+\rho_a)&0&\beta_a\\
0&0&\rho_a&-({b_a+d_a+\gamma_a})&0\\
0&\varepsilon_h&0&\varepsilon_a&-{\mu_m }\\
\end{bmatrix}$$\\

$$AY=\begin{bmatrix}
-(b_h+{\rho }_h)e_h+\beta_hm\\~\\
{\rho }_he_h-\left(b_h+\gamma_h\right)i_h\\~\\
-(b_a+{\rho }_a)e_a+\beta_am\\~\\
{\rho }_a e_a-\left(b_a+d_a+\gamma_a\right)i_a\\~\\
{\varepsilon }_hi_h+{\varepsilon }_ai_a-{\mu_m }m
\end{bmatrix}$$

\begin{equation}
\widehat{ G}(X,Y)= \begin{bmatrix} \widehat {G}_1(X,Y)\\ \widehat {G}_2(X,Y)\\ \widehat {G}_3(X,Y)\\ \widehat {G}_4(X,Y)\\ \widehat {G}_5(X,Y)\\  \end{bmatrix}=\begin{bmatrix}
\left(1-\frac{s_h}{1+m}\right)\beta_hm\\
0\\
\left(1-\frac{s_a}{1+m}\right)\beta_am-d_ae_ai_a\\
-d_ai_a^2\\0\\

\end{bmatrix}
\end{equation}
From $(9)$, $\widehat{ G}(X,Y)\geq0$ if and only if $d_a=0$. Thus, $E_0$ may not be globally asymptotically stable for some parameter values.

\section{Bifurcation analysis}\label{app:bifurcation}

The center manifold theory introduced by Castillo-Chavez and Song \cite{bib35} establishes the endemic equilibrium's stability. We employ the center manifold theory to examine the specific type of bifurcation exhibited by equation (4). The existence of forward bifurcation means that the disease-free and endemic equilibrium states are locally asymptotically stable if $R_0<1$ and $R_0>1$  respectively. Using the center manifold theory [35], a forward bifurcation occurs at bifurcation parameter $\phi=0$, if the coefficient constants $p<0$  and $ q>0 $, otherwise there is a backward bifurcation.
To apply the center manifold theory, a variable transformation is carried out on the normalized helminth model \eqref{fracequations}.

Let  $x_1=s_h,x_2=e_h,x_3=i_h,x_4=s_a,x_5=e_a,x_6=i_a,x_7=m$. Using the vector notation $X=(x_1,x_2,x_3,x_4,x_5,x_6,x_7)^T$ then \eqref{fracequations} can be written in the form 
$\dfrac{dX}{dt}=F(x)$ with $F=(f_1,f_2,f_3,f_4,f_5,f_6,f_7)$.

Choosing $\beta^*=\beta_h$  as the bifurcation parameter and solving for $p$ and $q$ at $R_0=1$ we have:
$$\beta^*=\frac{\mu_m\omega_h\theta_h\omega_a\theta_a-\beta_a\varepsilon_a\rho_a\omega_h\theta_h}{\varepsilon_h\rho_h\omega_a\theta_a}$$
with 

\begin{equation}
\begin{cases}
f_1&=b_h+{\gamma }_h x_3-\left({b_h+\dfrac{\beta^* x_7}{1+x_7 }}\right)x_1\\~\\
f_2&=\dfrac{\beta^* x_7}{1+x_7}x_1-\left(b_h+{\rho }_h\right)x_2\\~\\
f_3&={\rho }_hx_2-\left(b_h+\gamma_h\right)x_3\\~\\
f_4&=b_a+{\gamma }_a x_6-\left({b_a+\dfrac{\beta_a x_7}{1+x_7 }-d_ax_6}\right)x_4\\~\\
f_5&=\dfrac{\beta_a x_7}{1+x_7}x_4-\left(b_a+{\rho }_a-d_ax_6\right)x_5\\~\\
f_6&={\rho }_a x_5-\left(b_a+d_a+\gamma_a-d_ax_6\right)x_6\\~\\
f_7&={\varepsilon }_hx_3+{\varepsilon }_ax_6-{\mu_m }x_7
\end{cases}
\end{equation}

The Jacobian matrix of our model equation  at the disease-free equilibrium state $E_0$  is given as follows:
\begin{equation}\label{jcob}
J_{E_0}=\scriptsize\begin{bmatrix}
-b_h&0&\gamma_h&0&0&0&-\beta_h\\
0&-(b_h+\rho_h)&0&0&0&0&\beta_h\\
0&\rho_h&-({b_h+\gamma_h})&0&0&0&0\\
0&0&0&-b_a&0&d_a+\gamma_a&-\beta_a\\
0&0&0&0&-(b_a+\rho_a)&0&\beta_a\\
0&0&0&0&\rho_a&-({b_a+d_a+\gamma_a})&0\\
0&0&\varepsilon_h&0&0&\varepsilon_a&-\mu_m\\
\end{bmatrix}
\end{equation}
The Jacobian matrix \eqref{jcob} has been proved using the Routh-Hurwitz criterion to have a simple zero eigenvalue and 
negative eigenvalues. The right and left eigenvectors, $w_i$ and $v_j$, associated with the 
Jacobian matrix \eqref{jcob} at $ R_0 = 1$ are given by $w=(w_1,w_2,w_3,w_4,w_5,w_6,w_7)$ and $v=(v_1,v_2,v_3,v_4,v_5,v_6,v_7),$ 
where
\begin{align*}
w_1=&\dfrac{1}{b_h}\left[\dfrac{\gamma_h\rho_h}{b_h+\gamma_h}-(b_h+\rho_h)\right]w_2=\frac{1}{b_h}\left[\dfrac{\psi_h\rho_h}{\omega_h}-\phi_h\right]w_2\\
w_2=& w_2>0 \; \text{is}\; \text{free}\\
w_3=&\dfrac{\rho_hw_2}{b_h+\gamma_h}=\dfrac{\rho_hw_2}{\omega_h}\\
w_4=&\dfrac{\beta_a(b_h+\rho_h)}{b_a\beta^*}\left[\dfrac{\rho_a}{(b_a+\rho_a)(b_a+d_a+\gamma_a)}-1\right]w_2=\dfrac{\beta_a\phi_h}{b_a\beta^*}\left[\dfrac{\rho_a}{\phi_a\omega_a}-1\right]w_2\\
w_5=&\dfrac{\beta_a(b_h+\rho_h)w_2}{\beta^*(b_a+\rho_a)}=\dfrac{\beta_a\phi_hw_2}{\beta^*\phi_a}\\
w_6=&\dfrac{\beta_a\rho_a(b_h+\rho_h)w_2}{\beta^*(b_a+\rho_a)(b_a+d_a+\gamma_a)}=\dfrac{\beta_a\rho_a\phi_hw_2}{\beta^*\phi_a\omega_a}\\
w_7=&\frac{(b_h+\rho_h)w_2}{\beta^*}=\dfrac{\phi_hw_2}{\beta^*}\\
v_1=&v_4=0\\
v_2=& \dfrac{\rho_h}{b_h+\rho_h}v_3=\dfrac{\rho_h}{\phi_h}v_3\\
v_3=& v_3>0 \;\text{is} \;\text{free}\\
v_5= &\dfrac{1}{\beta_a}\left[\dfrac{\beta^*\rho_h}{b_h+\rho_h}-\dfrac{\mu_m(b_h+\gamma_h)}{\varepsilon_h}\right]v_3=\frac{1}{\beta_a}\left[\frac{\beta^*\rho_h}{\phi_h}-\frac{\mu_m\omega_h}{\varepsilon_h}\right]v_3\\
v_6= & \dfrac{(b_a+\rho_a)}{\rho_a\beta_a}\left[\dfrac{\beta^*\rho_h}{b_h+\rho_h}-\frac{\mu_m(b_h+\gamma_h)}{\varepsilon_h}\right]v_3=\frac{\phi_a}{\rho_a\beta_a}\left[\frac{\beta^*\rho_h}{\phi_h}-\frac{\mu_m\omega_h}{\varepsilon_h}\right]v_3\\
v_7= &\dfrac{b_h+\gamma_h}{\varepsilon_h}v_3=\frac{\omega_h}{\varepsilon_h}v_3.
\end{align*}
On the other hand,
\begin{align*}
\dfrac{\partial^2f_1}{\partial x_1\partial x_7}=&-\beta^*\\
\dfrac{\partial^2f_2}{\partial x_1\partial x_7}=&\beta^*\\
\dfrac{\partial^2f_4}{\partial x_4\partial x_7}=&-\beta_a\\
\dfrac{\partial^2f_4}{\partial x_4\partial x_6}= & d_a\\
\frac{\partial^2f_5}{\partial x_5\partial x_6}=d_a\\
\frac{\partial^2f_5}{\partial x_4\partial x_7}=&\beta_a\\
\dfrac{\partial^2f_6}{\partial x_6^2}= &2d_a\\
\dfrac{\partial^2f_1}{\partial x_7\partial \beta^*}= &-1\\
\dfrac{\partial^2f_2}{\partial x_7\partial \beta^*}= &1
\end{align*}
Now,
$$p=\sum_{k,i,j=1}^7v_kw_iw_j\dfrac{\partial^2f_k}{\partial x_i\partial x_j}(E_0,\beta^*)$$
implies
\begin{align*}
p=&v_2\left(w_1w_7\frac{\partial^2f_2}{\partial x_1\partial x_7}(E_0,\beta^*)\right)\\&+v_5\left(w_5w_6\frac{\partial^2f_5}{\partial x_5\partial x_6}(E_0,\beta^*)+w_4w_7\frac{\partial^2f_5}{\partial x_4\partial x_7}(E_0,\beta^*)\right)+v_6w_6^2\frac{\partial^2f_6}{\partial x_6^2}(E_0,\beta^*)
\end{align*}
By substitution,
\begin{align*}
p=&\rho_hv_3w_2^2\left[\frac{1}{b_h}\left(\frac{\gamma_h\rho_h}{\omega_h}-\phi_h\right)\right]\\&+\frac{2\beta_a\rho_a\phi_hd_av_3w_2^2}{\beta^{*2}\phi_a\omega_a^2\varepsilon_h}\left[\beta_h\rho_h\varepsilon_h-\mu_m\omega_h\phi_h\right]\\&+\frac{\beta_a\phi_hw_2^2v_3}{b_a\beta^{*2}\phi_a^2\omega_a\varepsilon_h}\left[(\beta_h\rho_h\varepsilon_h-\mu_m\phi_h\omega_h)(b_a\rho_ad_a+\phi_a(\rho_a-\phi_a\omega_a))\right]
\end{align*}
Similarly,
$$q=\sum_{k,i,=1}^7v_kw_i\dfrac{\partial^2f_k}{\partial x_i\partial \beta^*}(E_0,\beta^*)=v_2w_7\dfrac{\partial^2f_2}{\partial x_7\partial \beta^*}(E_0,\beta^*)$$
implying that
$$q=\dfrac{\rho_hv_3w_2}{\beta^*}>0.$$

The coefficient $q$ is positive. By Castillo-Chavez and Song \cite{bib35}, coefficient $p$ decides the local dynamics of endemic equilibrium. Therefore, if $p<0$ then system (4) will exhibit forward bifurcation; if $p>0$, it undergoes backward bifurcation.
\end{appendices}


\bibliographystyle{plain}
\bibliography{paper}

\end{document}